\documentclass[a4paper,conference]{IEEEtran}
\IEEEoverridecommandlockouts

\usepackage{amsthm}
\usepackage{mathrsfs}
\usepackage{amssymb}
\usepackage[cmex10]{amsmath}
\usepackage{stfloats}
\usepackage{graphicx}
\usepackage{epstopdf}

\usepackage{subfigure}
\usepackage{tabularx}
\usepackage{epsfig,epsf,color,balance,cite}
\usepackage{verbatim}
\usepackage{url}
\usepackage{bm}
\usepackage{caption}
\usepackage{threeparttable}
\usepackage{diagbox}
\usepackage{multirow}
\usepackage{stfloats}
\newtheorem{theorem}{Theorem}

\usepackage{algorithm}
\usepackage{algorithmic}
\makeatletter

\newcommand{\Rmnum}[1]{\expandafter\@slowromancap\romannumeral #1@}
\makeatother

\usepackage{geometry}
\usepackage{color}
\definecolor{mybackground}{RGB}{204,232,207}

\begin{document}
	\title{Resource Allocation for Semantic-Aware Mobile Edge Computing Systems}

     \author{
     	\IEEEauthorblockN{Yihan Cang, Ming Chen, \emph{Member, IEEE}, Zhaohui Yang, \emph{Member, IEEE}, Yuntao Hu,\\ Yinlu Wang, Zhaoyang Zhang, \emph{Senior Member, IEEE}, and Kai-Kit Wong, \emph{Fellow, IEEE}}
     	\thanks{Y. Cang, M. Chen, Y. Hu and Y. Wang are with the National Mobile Communications Research Laboratory, Southeast University, Nanjing 210096, China (e-mails: yhcang@seu.edu.cn, chenming@seu.edu.cn,  huyuntao@seu.edu.cn, yinluwang@seu.edu.cn).  Ming Chen is also with the Purple Mountain Laboratories, Nanjing 211100, China.}
     	\thanks{Z. Yang and Z. Zhang are with College of Information Science and Electronic Engineering, Zhejiang University, Hangzhou 310027, China, and with International Joint Innovation Center, Zhejiang University, Haining 314400, China, and also with Zhejiang Provincial Key Laboratory of Info. Proc., Commun. \& Netw. (IPCAN), Hangzhou 310027, China (e-mails: yang\_zhaohui@zju.edu.cn,   ning\_ming@zju.edu.cn).}
     	\thanks{K. Wong is with the Department of Electronic and Electrical Engineering, University College London, London, UK (e-mail:
     		kai-kit.wong@ucl.ac.uk).}
     }

      \maketitle

	\begin{abstract}
    In this paper, a semantic-aware joint communication and computation resource allocation framework is proposed for  mobile edge computing (MEC) systems. 
    In the considered system, each terminal device (TD) has a computation task, which needs to be executed by offloading to the MEC server.
    To further decrease the transmission burden, each TD sends the small-size extracted semantic information of tasks to the server instead of the  large-size raw data.
    An optimization problem of joint semantic-aware division factor, communication
    and computation resource management is formulated.
    The problem aims to minimize the maximum execution delay of all TDs 
    while satisfying energy consumption constraints.
    The original non-convex problem is transformed into a convex one based on the geometric programming and the optimal solution is obtained  by the alternating optimization algorithm. Moreover, the closed-form optimal solution of the semantic extraction factor is derived.  
    Simulation results show  that the proposed algorithm yields up to $37.10\%$ delay reduction compared
    with the benchmark algorithm without semantic-aware allocation. Furthermore,  small semantic extraction factors are preferred in the case of large task sizes and poor channel conditions.   
	\end{abstract}
	
	\begin{IEEEkeywords}
		Mobile edge computing, semantic-aware, compute-then-transmit,  resource management.
	\end{IEEEkeywords}

	\section{Introduction}
Due to the rapid development of  Internet of things, the desire of numerous terminal devices for a huge amount of emerging computation services has drastically increased the traffic of core networks. 
In order to release the burden of core networks, mobile edge computing (MEC),  deploying the network functions at the network edge, provides users with nearby real-time computing services \cite{7931566}. Therefore, as a new network architecture, MEC is considered as a promising paradigm to remedy the problem of computational resources and energy shortage of mobile equipments in future wireless networks\cite{9113305}. 
In MEC networks, communication and computation resources can be jointly optimized to improve a certain system utility, such as energy consumption \cite{8264794}, delay \cite{9613307}, throughput \cite{8249785}, computation efficiency \cite{8986845,9771419}, etc.
Up to now, resource management problems for MEC systems have been widely investigated. In \cite{7553459}, the authors researched on the  energy-efficient computation offloading mechanisms for heterogeneous MEC networks, considering the energy consumption of both task computing and wireless transmission.  The work in \cite{8249785} proposed a Nash bargaining resource allocation game to achieve maximum network throughput while guaranteeing each mobile user’s minimum rate requirements.

Motivated by Shannon's classic information theory, existing works including \cite{8264794,9613307,8249785,8986845,7553459} are dedicated to research on data-oriented communications. As a novel paradigm that involves the 
meaning of messages in communication, semantic communications, which concentrate
on transmitting the meaning of source information, have revealed the significant potential to
reduce the network traffic and thus alleviate spectrum shortage\cite{9530497,10024766}. 
Different from conventional communications, semantic information needs to be extracted from raw data before transmitted in semantic-aware networks. The development of  semantic information discipline provides a foundation for semantic communications\cite{carnap1952outline,6004632}. 
Nowadays, the realization of semantic communications has been demonstrated under different types of transmitted contents such as text \cite{9398576}, image \cite{8723589}, speech \cite{9450827} transmissions, etc. These works studied the feasibility of semantic communications under different scenarios from the viewpoint of   technical levels. However, the problem of how to implement resource management in a semantic communication system also needs to be investigated so as to explore the potential of practical semantic communication networks. 
Moreover, when the network loads are heavy,  communication loads can be converted to computation amounts with the help of semantic aware technology in MEC systems. Thus a  tradeoff between communications and computations can be achieved, improving the quality of service (QoS) of terminal devices.

Motivated by above observations, we attempt to investigate a novel resource management framework, which flexibly  orchestrates  communication and computation resources for semantic-aware MEC systems. 
The main contributions of this paper are summarized as follows:
\begin{itemize}
	\item We propose a semantic-aware  MEC system to achieve
	low execution latency. In the considered system,  each terminal device (TD) has a computation task  to be executed through offloading to the MEC server. To further decrease the transmission burden, each TD sends the small-size extracted semantic information of tasks to the server instead of the   large-size raw data.  With the help of semantic extraction, the amount of data uploading to the MEC server is  reduced, and thus lower transmission delay is achieved.
	\item To coordinate the operations on TDs and MEC server for minimizing the maximum execution delay of all users, the problem
	of joint communication and computation resource management is studied. In particular, this problem is
	constructed as an optimization framework aiming to acquire the optimal local computing rate and transmit power,
	remote computing capacity, and semantic extraction factors while satisfying the energy consumption constraints. Besides, this optimization framework can be applied to general semantic tasks. 
	\item Since the variables are highly coupled with each other, the formulated problem is non-convex which makes it intractable. To this end, a geometric programming based resource management algorithm is proposed to transform the original problem into a convex one and the optimal solution is obtained by the alternating optimization algorithm.  
 	Simulation results verify the outstanding performance of the proposed algorithm in terms of execution delay.  Compared with the benchmark algorithm without semantic-aware, the proposed algorithm can reduce up to $37.10\%$ delay reduction.
\end{itemize}

The rest of this paper is structured as follows. Section II elaborates on the system model and
problem formulation. In Section III, the optimal joint optimization algorithm is proposed.  Simulation results
are represented in Section V. Finally, Section VI draws the conclusions.

	\section{System Model and Problem Formulation}
	\subsection{System Model}

	\begin{figure}[t]
	\centering
	\includegraphics[width=0.5\textwidth]{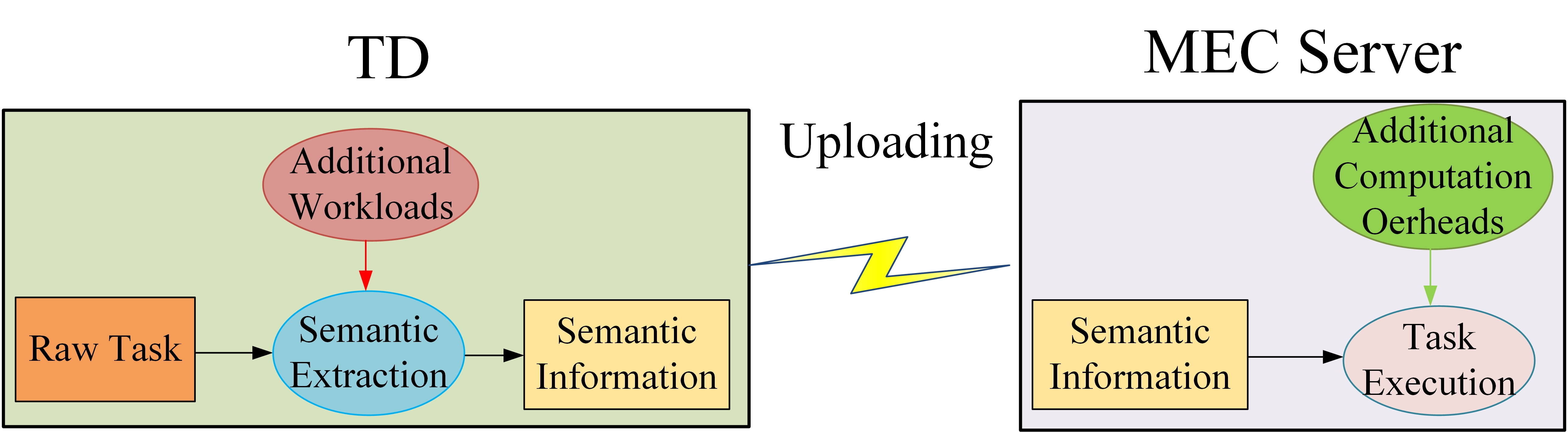}
	\caption{Flow chart of the semantic-aware MEC networks.} \label{flow chart}
\end{figure}
	Consider a semantic-aware MEC network which consists of a set $\mathcal{N}$ of $N$ TDs  and a MEC server attached to a base station (BS). Suppose that TD $n$ has a task with $A_n$ bits to be executed.  All the TDs  wirelessly access the MEC server for task offloading.  
	Assume that the MEC server is equipped with $N$ core CPUs such that offloaded tasks from different TDs can be executed in parallel.

	To further decrease the transmission burden, each TD equipped with a semantic processing unit sends the  extracted semantic information of tasks to the server instead of the raw data. With the help of semantic extraction, the amount of data uploading to the MEC server is reduced, thus releasing the traffic loads.  Meanwhile, additional workloads for semantic extraction are brought to each TD and the computation intensity of tasks gets large for the MEC server in order to process semantic information, as demonstrated in Fig.~\ref{flow chart}.  Compute-then-transmit protocol is adopted.  Specifically, TDs extract and transmit the semantic information of raw data to the server. For examples in Fig.~\ref{semantic}, we can utilize the picture cutout method for image transmission, where a pre-trained deep neural network is utilized to cut out the part of important character while discarding the irrelevant background\cite{9274895}. For text transmission, we can discard meaningless words and use abbreviations to replace the raw text without losing valid information by constructing knowledge graphs\cite{9416312}. In the considered scenario, only simple background knowledge is required to implement semantic extraction.  Hence, the pre-trained model is lightweight and universal.   Compared with task  processing,  the computation resource required for semantic extraction is much less \cite{10032275}.   
		
	\begin{figure}[t]
		\centering
		\includegraphics[width=0.48\textwidth]{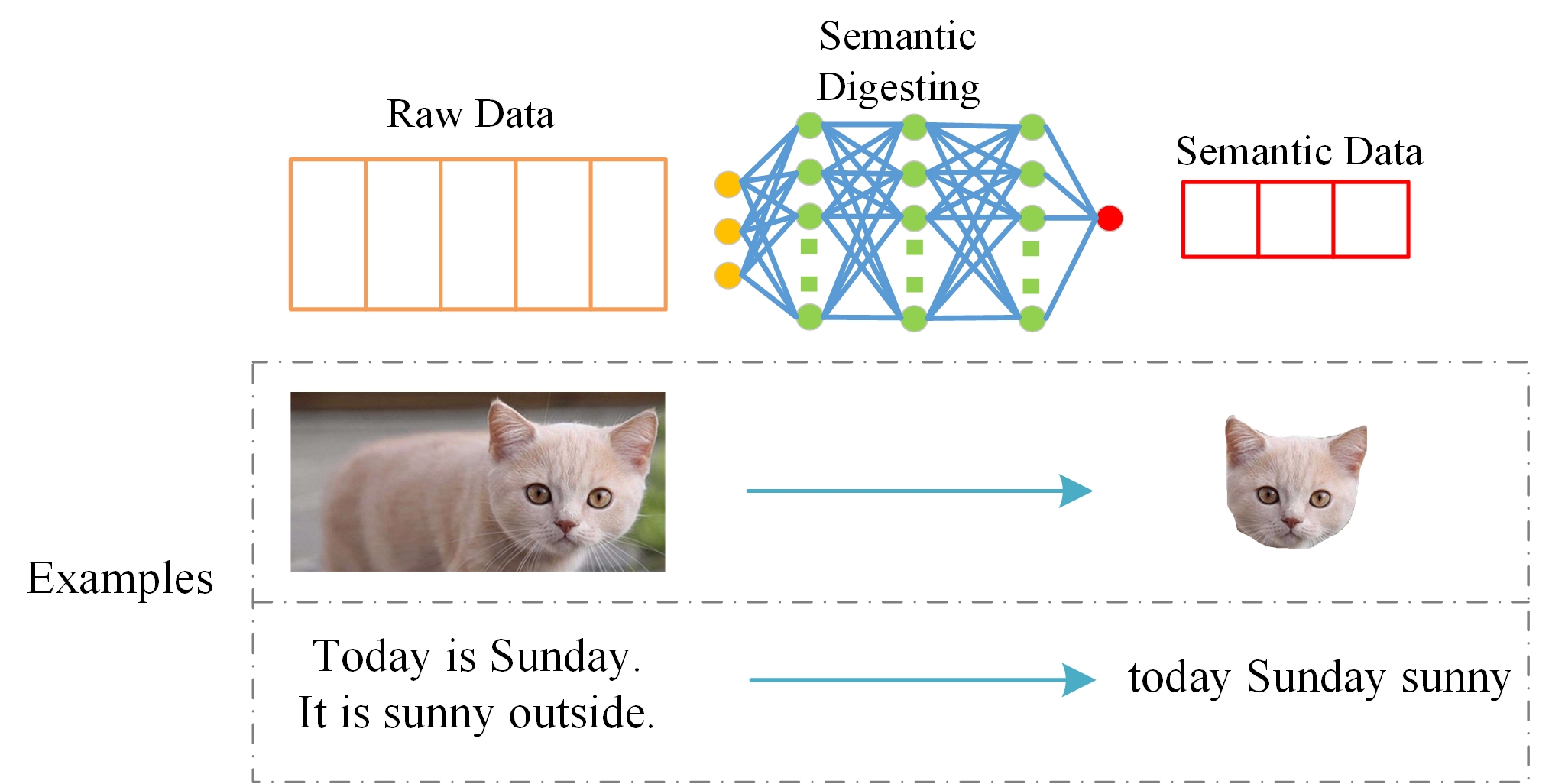}
		\caption{The schematic  diagram of semantic extraction and two examples.} \label{semantic}
	\end{figure}
	
	 \begin{figure} 
		\centering 
		\subfigure[Accuracy versus extraction factor.]{\label{}
			\includegraphics[width=0.47\linewidth]{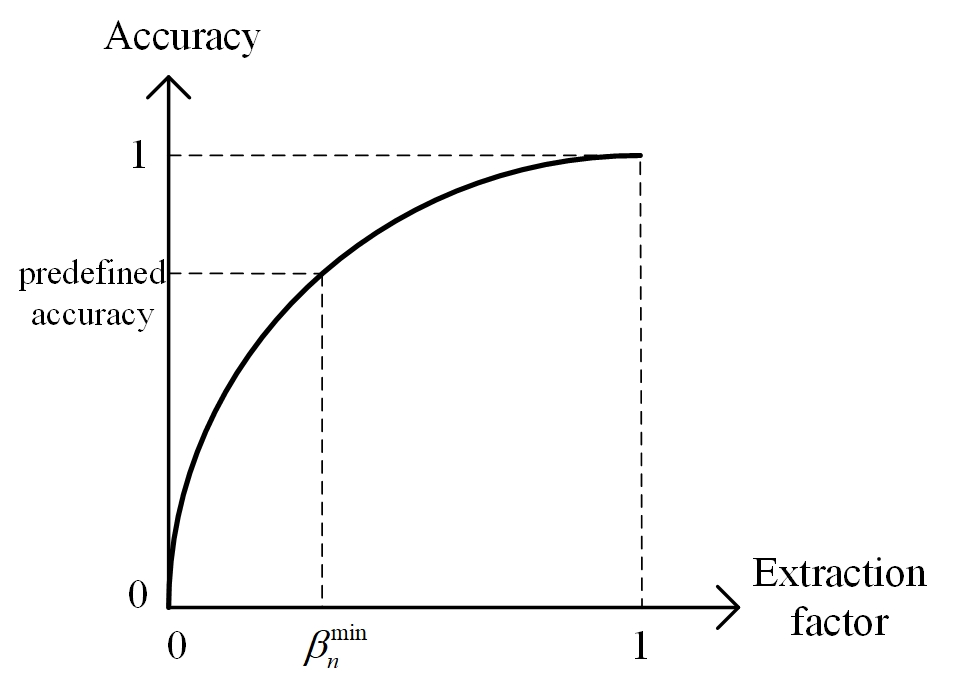}}
		\hspace{0.01\linewidth}
		\subfigure[Additional computation overhead  versus extraction factor.]{\label{}
			\includegraphics[width=0.4\linewidth]{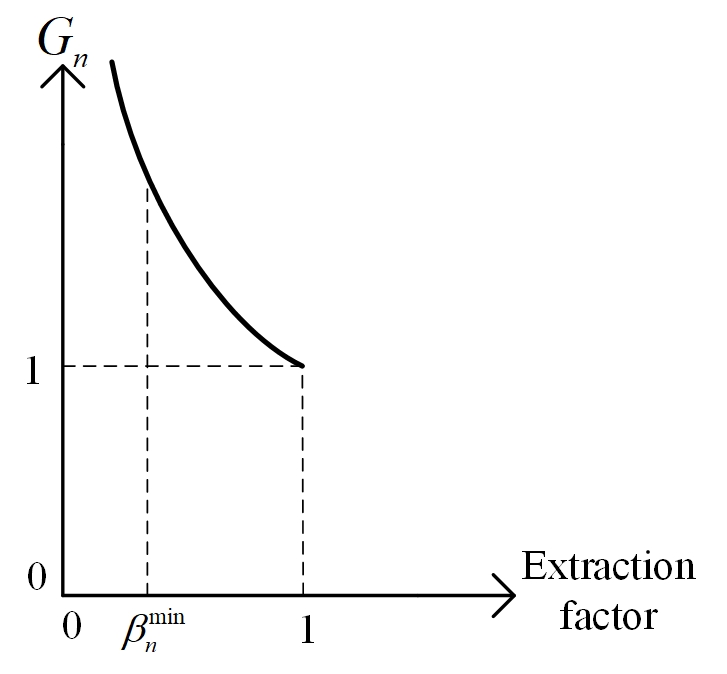}}
		\caption{Performance demonstration of semantic aware model.}
		\label{Semantic fig}
	\end{figure}

	Denote $\beta_n\in[\beta_n^{\min},1]$ as the extracting factor of the raw data to be uploaded by TD $n$, where $\beta_n^{\min}$ is the minimum extracting factor to maintain the most important information from TD $n$. As shown in Fig.~\ref{Semantic fig}(a)\cite{9398576,8723589,9274895}, the task accuracy of the computed results decreases when the extraction factor becomes small. In order to guarantee the predefined accuracy, $\beta_n$ should not be smaller than $\beta_n^{\min}$. Therefore, the amount of semantics to be uploaded for TD $n$ is equivalent to $A_n\beta_n$  and the transmission delay is given by
	\begin{align}
		T_n^T= \frac{A_n\beta_n}{R_n} ,\quad\forall n\in\mathcal{N},
	\end{align}
where $R_n$ denotes the achievable rate for TD $n$. Orthogonal frequency division multiple access (OFDMA) is adopted. Total bandwidth is  divided into $N$ sub-channels each with a bandwidth of $B$ (Hz) for each TD. Denote the channel gain between TD $n$ and MEC server by $h_{n}$, which keeps stable during the task offloading period. Therefore, $R_n$ is given as 
\begin{align} \label{R_n^U}
	R_{n}=B\log_2\left (1+\frac{h_{n} p_{n}^T }{\sigma^2}\right),\quad \forall n \in\mathcal{N}.  
\end{align}
where $\sigma^2$ denotes the additive noise power at the MEC server and $p_{n}^T$ represents the transmit power of TD $n$. 
	
	As shown in Fig.~\ref{flow chart}, at the transmitter, semantic extraction brings additional workloads. Meanwhile, at the receiver, additional computation overhead is required to process the semantic data rather than raw data. This indicates that  communication loads are converted to computation amounts through semantic transmission. 
	Denote $C_n$ as the workloads for semantic extraction in CPU cycles.  Note that $C_n$ should increase with the amount of extracted raw data $A_n$, while decreases with extraction factor $\beta_n$. Without loss of generality, we assume that\footnote{Note that the additional workload formulation for semantic extraction is provided in (3) as an example, which can be verified via simulations. Our method can also be applied to different types of workload formulation. }  
	\begin{align} \label{C_n(t)}
		C_n=\frac{a A_n}{\beta_n^k},\quad\forall n\in\mathcal{N},
	\end{align}
 	where $a >0, k> 0$ are constants relevant to the tasks. Therefore, the additional delay caused by semantic extraction in the local is expressed as
 	\begin{align}
 		T_n^L=\frac{C_n}{f_n^L}=\frac{aA_n}{\beta_n^kf_n^L},\quad\forall n\in\mathcal{N},
 	\end{align}
	where $f_n^L$ represents the local CPU frequency of TD $n$. Therefore, the local computing power consumption is given as 
	\begin{align}
		p_n^C=\kappa_n (f_n^L)^3,\quad \forall n \in\mathcal{N}, 
	\end{align}	
	where $\kappa_n$ denotes the energy coefficient of TD $n$.

	Denote $I_n$ as the computation intensity of TD $n$ in units of CPU cycles per raw data bit. Since only extracted semantic data of the raw data can be received by the MEC server, the computation intensity of the extracted semantic increases accordingly. As shown in Fig.~\ref{Semantic fig}(b), $G_n\in  [1,+\infty]$, which denotes the ratio of computation intensity of semantic data to that of raw data about TD $n$, increases monotonously with extraction factor $\beta_n$ getting small. The increase is caused by computations for processing semantic data and compensations for enhancing accuracy.  Besides, we should have  $G_n=1$ for the case that raw data are processed, i.e., $\beta_n=1$. Without loss of generality, we assume that
	\begin{align} \label{G_n}
		G_n=\frac{1}{\beta_n^p}, \quad\forall n\in\mathcal{N},
	\end{align}
	where $p>0$ is a constant relevant to specific task computation. Thus the remote semantic recovery and tasks execution delay can be given by
	\begin{align} \label{R_n^M(t)}
	  T_n^O=\frac{A_n\beta_nI_nG_n}{f_n^O},\quad \forall n\in\mathcal{N}, 
	\end{align}
	where $f_n^O$ denotes the computation resource of MEC server allocated to the TD $n$. 	Note that the precise formulation parameters of $C_n$ and $G_n$ can be obtained by fitting 
	the corresponding data points of abundant prior experiments. We will show that the proposed framework is well applicable to the general formulations of $C_n$ and $G_n$ in the simulations.
	
	\subsection{Problem Formulation}
 Due to the high transmit power of the BS, it is assumed that the
 downlink transmission time is negligible.  In this case, the total delay for TD $n$, which contains delay for semantic extraction, transmission delay,  and remote processing delay, is given as 
 \begin{align}
 	T_n=T_n^L+T_n^T+T_n^O,\quad\forall n\in\mathcal{N}. 
 \end{align}

	Let $\pmb{\Phi}=\{\pmb{f}^L,\pmb{p}^U,\pmb{f}^O, \pmb{\beta}\}$, where   $\pmb{f}^L=[f_1^L,\cdots,f_N^L]^T$ and $\pmb{f}^O=[f_1^O,\cdots,f_N^O]^T$ respectively denote local and remote computation capacity vectors, $\pmb{p}^T=[p_1^T,\cdots,p_N^T]^T$ is the  transmission power vector,   $\pmb{\beta}=[\beta_1,\cdots,\beta_N]^T$ denotes extraction factor vector of all TDs.   
	We aim at minimizing the maximum delay of all TDs under the energy consumption of each TD which consists of semantic extraction and transmission energy consumption. 
	Mathematically, the optimization problem is posed as 	
	\begin{subequations} \label{P1}
		\begin{align}
			\min_{\pmb{\Phi}}\ \   &\max_{n\in\mathcal{N}}\ T_n,\\
			\textrm{\textrm{s.t.}}\ \ 
			& p_n^CT_n^L +p_n^T T_n^T\leq E_n,  \quad \forall n\in\mathcal{N},\\
			&0\leq f_n^L\leq f_n^{\max}, \quad \forall n\in\mathcal{N},\\
			&\sum_{n=1}^N f_n^O\leq F_{MEC},\\
			&0\leq p_{n}^T\leq p_{n}^{\max}, \quad \forall n\in\mathcal{N},\\ 
			&\beta_n^{\min}\leq\beta_n\leq 1, \quad \forall n\in\mathcal{N}, 
		\end{align}
	\end{subequations}
	where $f_n^{\max}$ and $p_{n}^{\max}$ respectively denote the maximum local computation capacity and maximum transmission power of TD $n$,  $F_{MEC}$ is the maximum computation capacity of MEC server, and $E_n$ represents the maximum energy consumption of TD $n$. 
	In problem \eqref{P1}, constraint  (\ref{P1}b) reflects that the energy consumption of TD $n$ should not be larger than its predefined energy consumption threshold. Constraints (\ref{P1}c) and (\ref{P1}e) enforce the local computation capacity and transmission power to be non-negative and should not exceed the predefined budget. The remote computation capacity allocation is constrained in (\ref{P1}d);  and constraint (\ref{P1}f) specifies the semantic extraction factor limitations. Different from previous works \cite{8264794,9613307,8249785,8986845,7553459}, semantic extraction factor $\beta_n$ is considered in problem \eqref{P1}. 
	
	\section{Optimal Algorithm Design}
	\subsection{Problem Transformation}
	The objective function (\ref{P1}a) is in a max-min form, which is complicated.  To handle this
	issue, we introduce an auxiliary variable $t$. Hence, problem \eqref{P1} is transformed into the following equivalent
	problem:
		\begin{subequations} \label{P2}
		\begin{align}
			\min_{\pmb{\Phi},t}\  &t,\\
			\textrm{\textrm{s.t.}}\ 
			&T_n^L+T_n^T+T_n^O\leq t, \quad \forall n\in\mathcal{N},\\
			&\text{(\ref{P1}b)}-\text{(\ref{P1}f)}. \nonumber
		\end{align}
	\end{subequations}
Due to the intractability complex term   $p_n^TT_n^T$ in constraint (\ref{P1}b), we introduce $t_n^T$ as the transmission delay for TD $n$  and  $\pmb{t}^T=[t_1^T,\cdots,t_N^T]^T$ as the transmission delay vector. Thus, problem \eqref{P2} is reformulated as
\begin{subequations} \label{P3}
	\begin{align}
		\min_{\pmb{\Phi},\pmb{t}^T,t}\  &t,\\
		\textrm{\textrm{s.t.}}\ 
		&\frac{aA_n}{\beta_n^kf_n^L}+t_n^T+\frac{A_n\beta_nI_n}{f_n^O\beta_n^p}\leq t, \quad \forall n\in\mathcal{N},\\
		& p_n^CT_n^L +p_n^T t_n^T\leq E_n,  \quad \forall n\in\mathcal{N},\\
		&t_n^TR_n\geq\beta_nA_n,  \quad \forall n\in\mathcal{N},\\
		&\text{(\ref{P1}c)}-\text{(\ref{P1}f)},\nonumber
	\end{align}
\end{subequations}
where constraint (\ref{P3}d) implies that all the extracted semantic information should be uploaded to the MEC server. Since constraint (\ref{P3}d) is non-convex, we  introduce   $e_n^T=p_n^Tt_n^T$ as the transmission energy consumption variable for TD $n$. Therefore, (\ref{P3}d) is transformed into 
\begin{align} \label{12}
	t_n^TB\log_2(1+\frac{h_ne_n^T}{t_n^T\sigma^2})\geq\beta_nA_n,\quad\forall n\in\mathcal{N}. 
\end{align}
Since function  $f(e_n^T)=B\log_2(1+\frac{h_ne_n^T}{\sigma^2})$ is concave with respect to $e_n^T$, its perspective $t_n^Tf(e_n^T/t_n^T)$ is also concave with $(e_n^T,t_n^T)$\cite{boyd2004convex}. Therefore, constraint \eqref{12} is a convex set. Moreover, due to the coupleness among $\beta_n$, $f_n^L$, and $f_n^O$, the geometric programming (GP) algorithm can be utilized to handle this issue. In specific, denote $\beta_n=e^{\tilde{\beta}_n}$, $f_n^L=e^{\tilde{f}_n^L}$, $f_n^O=e^{\tilde{f}_n^O}, (\forall n\in\mathcal{N})$. Problem \eqref{P3} is transformed into the following problem
		\begin{subequations} \label{P4}
		\begin{align}
			&\min_{\pmb{\tilde f}^L,\pmb{e}^T,\pmb{t}^T,\pmb{\tilde f}^O, \pmb{\tilde \beta},t}\  t,\\
			\textrm{\textrm{s.t.}}\  
			&aA_ne^{-k\tilde\beta_n-\tilde f_n^L}+  t_n^T  + A_nI_ne^{(1-p)\tilde\beta_n-\tilde f_n^O}  \leq t, \forall n\in\mathcal{N},\\
			& aA_n\kappa_n e^{2\tilde f_n^L-k\tilde\beta_n} + e_n^T\leq E_n,  \quad \forall n\in\mathcal{N},\\
			&	t_n^TB\log_2(1+\frac{h_ne_n^T}{t_n^T\sigma^2})\geq e^{\tilde\beta_n}A_n,\quad\forall n\in\mathcal{N},\\
			&\tilde f_n^L\leq \ln(f_n^{\max}), \quad \forall n\in\mathcal{N},\\
			&\sum_{n=1}^N e^{\tilde f_n^O}\leq F_{MEC},\\
			&0\leq e_{n}^T\leq p_{n}^{\max}t_n^T, \quad \forall n\in\mathcal{N},\\ 
			&\ln(\beta_n^{\min})\leq\tilde\beta_n\leq 0, \quad \forall n\in\mathcal{N}, 
		\end{align}
	\end{subequations}
which is  convex. 
		\begin{theorem}
		For problem \eqref{P4}, we have the following properties:
		\begin{enumerate}
			\item[a)] Problem \eqref{P4} is equivalent to problem \eqref{P1}.
			\item[b)] For the optimal solution of problem \eqref{P4}, the equality  in (\ref{P4}b) holds for $\forall n\in\mathcal{N}$. 
			\item[c)] The equality in (\ref{P4}e) holds for the optimal solution of problem \eqref{P4}. 
		\end{enumerate}
	\end{theorem}
	\begin{proof}
		 a) can be obtained according to the transformation procedures from \eqref{P2} to \eqref{P4}.  b) and c) can be derived by contradiction. Detailed proof is omitted due to space limitations. 
	\end{proof}

\subsection{Optimal Algorithm}
To further decrease the  complexity of solving problem \eqref{P4} with utilizing the interior-point method, an alternating optimization algorithm is proposed. In the proposed alternating algorithm, three subproblems requires to be solved, i.e., local and remote computation rate optimization, transmission delay and energy consumption optimization, as well as semantic extraction factor optimization. 
\subsubsection{Optimal Local and Remote Computation Rate}
According to (\ref{P4}b), the optimal $t$ decreases with $\tilde f_n^L$. Since the left hand side (LHS) of (\ref{P4}c) increases with $\tilde f_n^L$, the optimal $\tilde f_n^L$ should satisfy either (\ref{P4}c) or (\ref{P4}e). Due to that $f_n^L=e^{\tilde f_n^L},\beta_n=e^{\tilde\beta_n}$ and setting (\ref{P4}c) with equality, the optimal $f_n^L$ is given by
\begin{align} \label{optlocal}
	f_n^L=\min\left\{f_n^{\max},\sqrt{\frac{\beta_n^k\left(E_n-e_n^T\right)}{aA_n\kappa_n}}\right\},\quad\forall n\in\mathcal{N}. 
\end{align}

	According to Theorem 1, setting (\ref{P4}b) with equality yields 
	\begin{align}
		f_n^O=\frac{A_nI_n\beta_n^{1-p}}{t-\frac{aA_n}{\beta_n^kf_n^L}-t_n^T}.\label{f_n^O}
	\end{align}
	Therefore, the bisection method can be adopted to find the optimal $t$ and $f_n^O$. The upper bound and lower bound of $t$ can be respectively provided as 
	\begin{align} 
		t^{\max}&=\max_{n\in\mathcal{N}}\frac{aA_n}{\beta_n^kf_n^L}+t_n^T+\frac{A_nI_n\beta_n^{1-p}N}{F_{MEC}},\label{tmax}\\
		t^{\min}&=\max_{n\in\mathcal{N}}\frac{aA_n}{\beta_n^kf_n^L}+t_n^T+\frac{A_nI_n\beta_n^{1-p}}{F_{MEC}}.\label{tmin}
	\end{align}
	The detailed steps are summarized in Algorithm~\ref{alg1}.
		\begin{algorithm}[t]    
		\algsetup{linenosize=\small}
		\small
		\caption{Optimal Local and Remote Computing Capacity  Algorithm}
		\begin{algorithmic}[1]         \label{alg1}
			\STATE {\bf Initialize:}  $t^{\max}$ and $t^{\min}$ respectively in \eqref{tmax} and \eqref{tmin}, bisection accuracy $\epsilon_1$. 
			
			\STATE Calculate the optimal local computation rate $f_n^L$ according to \eqref{optlocal}.
			\STATE  {\bf Repeat:}
			\STATE   \hspace*{0.2in}  Set $t \leftarrow \frac{t^{\max}+t^{\min}}{2}$;
			\STATE \hspace*{0.2in} {\bf For} $n=1:N$:
			\STATE \hspace*{0.4in} Calculate $f_n^O$ according to \eqref{f_n^O}.\\
			
			\STATE \hspace*{0.2in} {\bf If} $\sum_{n=1}^N f_n^O(t)- F_{MEC} < 0$:\\
			\STATE\hspace*{0.4in} $t^{\max}\leftarrow t$;\\
			\STATE\hspace*{0.2in} {\bf Else:}\\
			\STATE\hspace*{0.4in} $t^{\min}\leftarrow t$.\\
			\STATE\hspace*{0.2in} {\bf End}\\
			\STATE{\bf Until:} $t^{\max}-t^{\min}\leq\epsilon_1$. 

			\STATE  {\bf Output:} optimal $f_n^L$ and $f_n^O$.   
		\end{algorithmic}
	\end{algorithm}
	
\subsubsection{Optimal Transmission Delay and Energy Consumption}
	With fixed $\pmb{\tilde f}^L$, $\pmb{\tilde f}^O$, and $\pmb{\tilde\beta}$, problem \eqref{P4} is reduced to the following $N$ parallel subproblems: 
	  		\begin{subequations} \label{P4.1}
	  	\begin{align}
	  		\min_{e_n^T,t_n^T}\quad  &t_n^T,\\
	  		\textrm{\textrm{s.t.}}\quad 
	  		&e_n^T \leq E_n-aA_n\kappa_n(f_n^L)^2/\beta_n^k, \\
	  		&	t_n^TB\log_2(1+\frac{h_ne_n^T}{t_n^T\sigma^2})\geq \beta_nA_n,\\
	  		&0\leq e_{n}^T\leq p_{n}^{\max}t_n^T. 
	  	\end{align}
	  \end{subequations}
	Similarly, bisection method can be used to solve problem \eqref{P4.1}. Specifically, given $t_n^T$, $e_n^T$ should satisfy the following condition according to (\ref{P4.1}b) and (\ref{P4.1}d):
	\begin{align} \label{19}
		0\leq e_n^T\leq \nu_n,
	\end{align}
	where $\nu_n=\min\left\{E_n-aA_n\kappa_n(f_n^L)^2/\beta_n^k,p_{n}^{\max}t_n^T\right\}$. If $t_n^TB\log_2(1+\frac{h_n\nu_n}{t_n^T\sigma^2})\geq \beta_nA_n$, problem \eqref{P4.1} is feasible; otherwise, it is infeasible. The detailed  procedures are summarized in Algorithm~\ref{alg2}.
	\begin{algorithm}[t]    
		\algsetup{linenosize=\small}
		\small
		\caption{Optimal Transmission Delay and Energy Consumption  Algorithm}
		\begin{algorithmic}[1]         \label{alg2}
			\STATE {\bf Initialize:}  $ub\leftarrow$sufficiently large, $lb\leftarrow$0, bisection accuracy  $\epsilon_2$. 
			
			\STATE  {\bf For} $n=1:N$:
			\STATE  \hspace*{0.2in}{\bf Repeat:}
			\STATE  \hspace*{0.4in}  Set $t_n^T \leftarrow \frac{ub+lb}{2}$;
			
			\STATE \hspace*{0.4in} {\bf If} $t_n^TB\log_2(1+\frac{h_n\nu_n}{t_n^T\sigma^2})\geq \beta_nA_n$:\\
			\STATE\hspace*{0.6in} $ub\leftarrow t_n^T$;\\
			\STATE\hspace*{0.4in} {\bf Else:}\\
			\STATE\hspace*{0.6in} $lb\leftarrow t_n^T$.\\
			\STATE\hspace*{0.4in} {\bf End}\\
			\STATE \hspace*{0.2in} {\bf Until:} $ub-lb\leq\epsilon_2$. \\
			\STATE {\bf End}\\

			\STATE  {\bf Output:} the optimal $t_n^T$ and $e_n^T$.   
		\end{algorithmic}
	\end{algorithm}

\subsubsection{Optimal Semantic Extraction Factor}
With fixed $\pmb{\tilde{f}}^L,\pmb{t}^T,\pmb{e}^T,\pmb{\tilde{f}}^O$, problem \eqref{P4} is reduced to 
	\begin{subequations} \label{P5}
	\begin{align}
		\min_{\pmb\beta,t}\quad  &t,\\
		\textrm{\textrm{s.t.}}\quad 
		&\frac{aA_n}{f_n^L\beta_n^k}+  t_n^T  + \frac{A_nI_n\beta_n^{1-p}}{f_n^O}  \leq t, \quad \forall n\in\mathcal{N},\\
		& \frac{aA_n\kappa_n(f_n^L)^2}{\beta_n^k} + e_n^T\leq E_n,  \quad \forall n\in\mathcal{N},\\
		&\beta_nA_n\leq t_n^TB\log_2\left(1+\frac{h_ne_n^T}{t_n^T\sigma^2}\right), \quad \forall n\in\mathcal{N},\\
		&\beta_n^{\min}\leq\beta_n\leq 1, \quad \forall n\in\mathcal{N}. 
	\end{align}
\end{subequations}
Problem \eqref{P5} can be further  equivalent to 
	\begin{subequations} \label{P5.1}
	\begin{align} 
		\min_{\beta_n}\quad  &\frac{aA_n}{f_n^L\beta_n^k}+  t_n^T  + \frac{A_nI_n\beta_n^{1-p}}{f_n^O},\\
		\textrm{\textrm{s.t.}}\quad 
		&\eta_1\leq\beta_n\leq \eta_2, 
	\end{align}
\end{subequations}
where $\eta_1=\max\left\{\beta_n^{\min},\sqrt[k]{\frac{aA_n\kappa_n(f_n^L)^2}{E_n-e_n^T}}\right\}$ and  $\eta_2=\min\left\{1,\frac{t_n^T}{A_n}B\log_2\left(1+\frac{h_ne_n^T}{t_n^T\sigma^2}\right)\right\}$. 
We consider the following two cases for parameter $p$.  

Case 1: $p\geq 1$. In this case,  (\ref{P5.1}a) decreases with $\beta_n$. Therefore, the optimal $\beta_n=\eta_2$.

Case 2: $0<p<1$. Since (\ref{P5.1}a) is convex with respect to $\beta_n$ and the stationary point can be derived by $\mu=\sqrt[k+1-p]{\frac{akf_n^O}{f_n^LI_n(1-p)}}$. Therefore, the closed-form optimal solution is given as
\begin{align} \label{22}
	\beta_n=\left\{\begin{aligned}
		&\eta_1,\quad\text{if }\mu\leq\eta_1,\\
		&\mu,\quad\ \text{if }\eta_1<\mu\leq\eta_2,\\
		&\eta_2,\quad\text{if }\eta_2<\mu.
	\end{aligned}\right.
\end{align}

\textbf{\emph{Remark:}} According to \eqref{22}, semantic extraction factor $\beta_n$ decreases with task size $A_n$ since the data amount to be uploaded can be greatly reduced with a small $\beta_n$.  Moreover, bad channel conditions results in a small $\beta_n$ as a small $\beta_n$ can alleviate the negative impacts incurred by bad channel conditions. Furthermore, according to \eqref{19}, energy consumption decreases with $\beta_n$ decreasing. This is due to the fact that transmit power is reduced with the aid of semantic transmission. 

\subsection{Performance Analysis}
The overall optimization algorithm for semantic-aware MEC systems is summarized in Algorithm~\ref{alg3}. Since problem \eqref{P4} is convex, the optimal solution can be obtained through alternately optimizing computation resource, communication resource, and semantic extraction factor. The complexity of the proposed algorithm lies in solving three subproblems. For local and remote computation rate optimization,
the complexity of the bisection method can be estimated as $\mathcal{O}\left(3N+N\log_2(1/\epsilon_1)\right)$, where $\mathcal{O}\left(\log_2(1/\epsilon_1)\right)$  is the average searching numbers of bisection method with accuracy $\epsilon_1$. For transmission delay and energy consumption optimization, the complexity is given as $\mathcal{O}\left(N\log_2(1/\epsilon_2)\right)$, where $\epsilon_2$ is the accuracy of the bisection method. Besides, the complexity of semantic extraction factor optimization is $\mathcal{O}(N)$ since the closed-form solution is given in \eqref{22}. In summary, the proposed algorithm can be implemented with linear complexity, which is significantly reduced compared with the  interior-point method. 
	  	\begin{algorithm}[t]    
		\algsetup{linenosize=\small}
		\small
		\caption{Overall Semantic-Aware MEC Resource Allocation Algorithm}
		\begin{algorithmic}[1]         \label{alg3}
			\STATE  {\bf Initialize:} $N$, $A_n$, $\kappa_n$, $f_n^{\max}$, $p_n^{\max}$, $\beta_n^{\min}$, $F_{MEC}$, $E_n$. \\
			\STATE {\bf Repeat:} \\
			\STATE\hspace*{0.2in}With given transmission delay, energy consumption and semantic extraction,  obtain local and remote computation rate according to Algorithm~\ref{alg1}. \\
			\STATE\hspace*{0.2in}With given local, remote computation rate, and semantic extraction,  obtain transmission delay and energy consumption according to Algorithm~\ref{alg2}.\\
			\STATE\hspace*{0.2in}Update  semantic extraction factors  through solving problem \eqref{P5.1}.\\
			\STATE {\bf Until:} the objective value (\ref{P4}a) converges.\\
			\STATE  {\bf Output:} the optimal $\pmb{f}^L$, $\pmb{e}^T$, $\pmb{t}^T$, $\pmb{f}^O$, $\pmb{\beta}$, and $t$.  
		\end{algorithmic}
	\end{algorithm}

\section{Simulation Results}
	\begin{table}[t]
	\caption{Simulation parameters} 
	\label{table} 
	\centering 
	\begin{tabular}{|m{2.5cm}<{\centering}|m{2.5cm}<{\centering}|m{2.5cm}<{\centering}|}
		\hline
		\hline
		$\sigma^2=-174$ dBm/Hz&$B=1$ MHz&$I_n=70$ cycles/bit\\
		\hline
		$\kappa_n=10^{-26}$ & $A_n=3$ Mbits &$f_n^{\max}=1$ GHz\\
		\hline
		$p_n^{\max}=1$ Watts&$\beta_n^{\min}=0.6$&$F_{MEC}=13$ GHz\\
		\hline
		$a=10^{-5}$ & $k=4$ &$p=3$\\
		\hline
		$E_n=0.5$ Watts&$\epsilon_1=10^{-7}$&$\epsilon_2=10^{-7}$\\		\hline
		\hline
	\end{tabular}
\end{table}
In this section, numerical results are conducted to evaluate the 
performance of the proposed framework.  There are $N=10$ TDs with each equipped with a semantic extracting processing unit to implement semantic-aware edge computing. The distances between TDs and BS are evenly set in $[120,255]$ meters. The channel model follows in \cite{9449944}.   
The other parameters are set as in Table~\ref{table} unless otherwise mentioned.

Fig.~\ref{fig_E} illustrates the  total delay performances versus the predefined energy threshold. 
The following benchmark algorithms are provided to compare with the proposed semantic-aware MEC framework: 1) MEC without semantic, i.e., $\beta_n(t)=1$ for all TDs. This algorithm corresponds to the conventional  MEC scenario that TDs directly upload raw data to the MEC server\cite{9613307}. 2) Local Execution Algorithm, where the tasks are executed locally under the energy and power constraints. 
As can be seen in Fig.~\ref{fig_E}, maximum delay of all users decreases with energy thresholds. Moreover, the proposed framework significantly outperforms benchmark algorithms under different tasks sizes. On average, the proposed framework yields up to $37.10\%$ and $69.35\%$ delay reduction compared with MEC without semantic algorithm and local execution algorithm, respectively. This is due to the fact that semantic-aware MEC can extract key data, thus efficiently reducing transmission delay while maintaining accuracy. 

To test the framework in applications to various semantic-aware scenarios,  in Fig.~\ref{fig_beta}, we plot the performance comparisons between
different workload expressions $C_n$ and additional computation overhead expressions $G_n$ under various minimum extraction
factors. Note that the case when $\beta_n^{\min} = 1$ is
equivalent to the traditional MEC without semantic-aware. As can be seen, as the minimum
extraction factor decreases from $1.0$ to $0.5$, the total delay becomes small. This can be explained by that a smaller extraction
factor significantly reduces the amount of data to be uploaded,
thus the transmit delay is reduced. 
Furthermore, the
performance of the proposed algorithm is similar even under different
workload expressions. This shows its ability to apply to different
semantic-aware scenarios. 
Besides, larger $a$ and $k$ can result in a larger delay. This is because as $a$ and $k$ become large, more workloads for semantic extraction shall be finished at TDs according to \eqref{C_n(t)}, which incurs a larger processing delay. 
Similarly, a larger $p$ leads to a higher delay in Fig.~\ref{fig_beta}.
This is due to the fact that $G_n$ becomes larger when $p$
gets bigger, which indicates that processing a bit of semantic
data requires more CPU cycles at the MEC server. Thus a higher execution delay is needed for semantic recovery.

	\begin{figure}[t]
	\centering
	\includegraphics[width=0.5\textwidth]{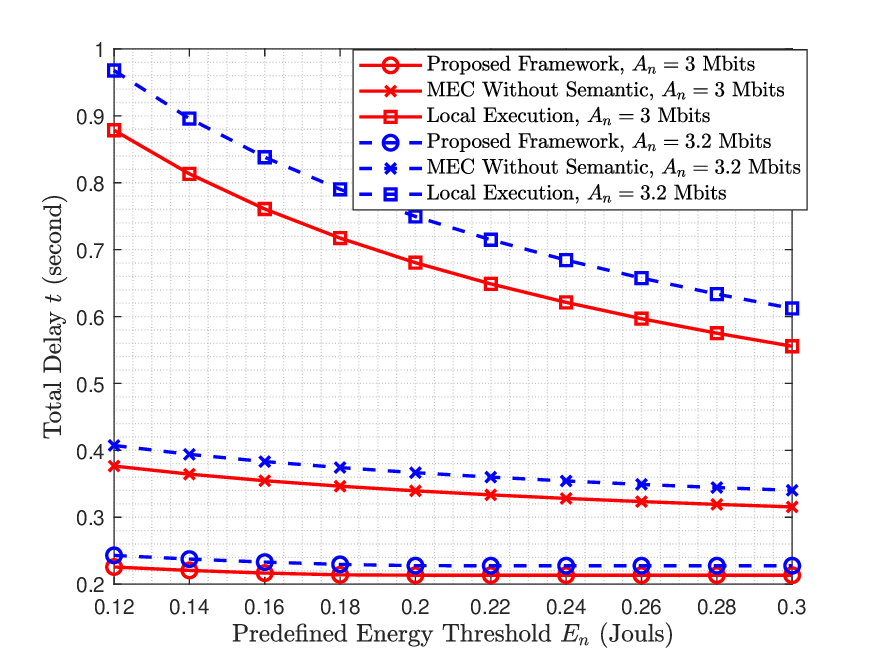}
	\caption{Performance comparisons between different algorithms under different energy thresholds and tasks sizes.} \label{fig_E}
\end{figure}
	
	\begin{figure}[t]
		\centering
		\includegraphics[width=0.5\textwidth]{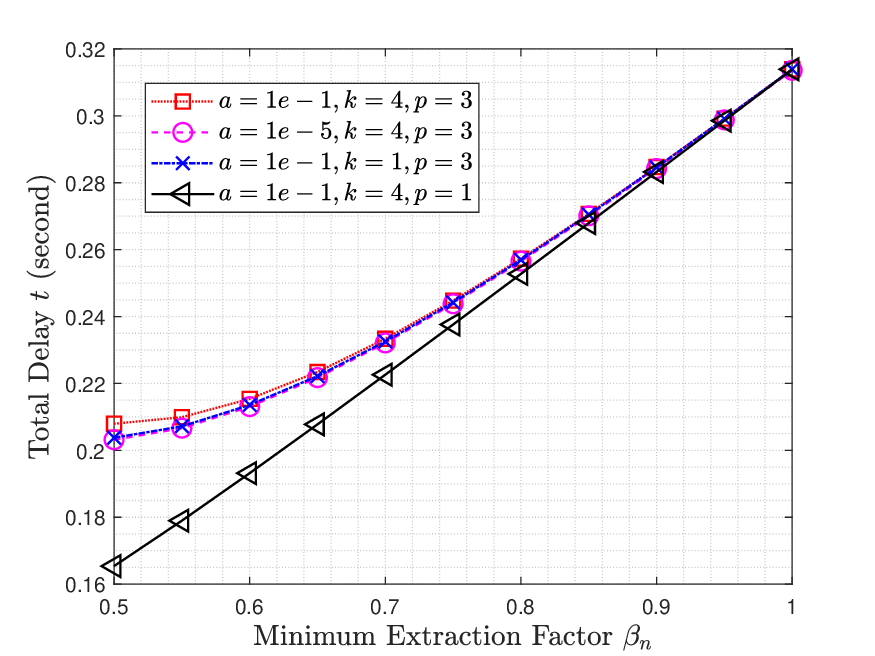}
		\caption{Performance comparisons between different $C_n$ and $G_n$ expressions under different minimum extraction factors.} \label{fig_beta}
	\end{figure}
	
\section{Conclusions}
In this paper, we have proposed a joint communication and
computation resource allocation framework for MEC systems with the aid of prevalent semantic transmission technology. An optimization problem has been formulated to optimize semantic extraction factors, communication
and computation resources to minimize the maximum delay of all TDs by taking into account energy consumption constraints. 
The original non-convex problem is transformed into a convex one based on GP, which can be efficiently solved through the proposed alternating   optimization algorithm. Simulation results demonstrate the superiority 
 of the proposed framework over the benchmark
schemes with respect to delay. Besides, the maximum delay of all users is significantly reduced as the semantic
extraction factor gets small.

	\bibliography{IEEEabrv,Ref}

\begin{thebibliography}{10}
\providecommand{\url}[1]{#1}
\csname url@samestyle\endcsname
\providecommand{\newblock}{\relax}
\providecommand{\bibinfo}[2]{#2}
\providecommand{\BIBentrySTDinterwordspacing}{\spaceskip=0pt\relax}
\providecommand{\BIBentryALTinterwordstretchfactor}{4}
\providecommand{\BIBentryALTinterwordspacing}{\spaceskip=\fontdimen2\font plus
\BIBentryALTinterwordstretchfactor\fontdimen3\font minus
  \fontdimen4\font\relax}
\providecommand{\BIBforeignlanguage}[2]{{%
\expandafter\ifx\csname l@#1\endcsname\relax
\typeout{** WARNING: IEEEtran.bst: No hyphenation pattern has been}%
\typeout{** loaded for the language `#1'. Using the pattern for}%
\typeout{** the default language instead.}%
\else
\language=\csname l@#1\endcsname
\fi
#2}}
\providecommand{\BIBdecl}{\relax}
\BIBdecl

\bibitem{7931566}
T.~Taleb, K.~Samdanis, B.~Mada, H.~Flinck, S.~Dutta, and D.~Sabella, ``{On
  Multi-Access Edge Computing: A Survey of the Emerging 5G Network Edge Cloud
  Architecture and Orchestration},'' \emph{{IEEE Commun. Surv. Tutor.}},
  vol.~19, no.~3, pp. 1657--1681, thirdquarter 2017.

\bibitem{9113305}
Q.-V. Pham, F.~Fang, V.~N. Ha, M.~J. Piran, M.~Le, L.~B. Le, W.-J. Hwang, and
  Z.~Ding, ``{A Survey of Multi-Access Edge Computing in 5G and Beyond:
  Fundamentals, Technology Integration, and State-of-the-Art},'' \emph{IEEE
  Access}, vol.~8, pp. 116\,974--117\,017, June 2020.

\bibitem{8264794}
X.~Hu, K.-K. Wong, and K.~Yang, ``{Wireless Powered Cooperation-Assisted Mobile
  Edge Computing},'' \emph{{IEEE Trans. Wirel. Commun.}}, vol.~17, no.~4, pp.
  2375--2388, Apr. 2018.

\bibitem{9613307}
Y.~Yu, Y.~Yan, S.~Li, Z.~Li, and D.~Wu, ``{Task Delay Minimization in Wireless
  Powered Mobile Edge Computing Networks: A Deep Reinforcement Learning
  Approach},'' in \emph{Proc. International Conference on Wireless
  Communications and Signal Processing (WCSP)}, Changsha, China, Oct. 2021, pp.
  1--6.

\bibitem{8249785}
Z.~Zhu, J.~Peng, X.~Gu, H.~Li, K.~Liu, Z.~Zhou, and W.~Liu, ``{Fair Resource
  Allocation for System Throughput Maximization in Mobile Edge Computing},''
  \emph{IEEE Access}, vol.~6, pp. 5332--5340, Jan. 2018.

\bibitem{8986845}
F.~Zhou and R.~Q. Hu, ``{Computation Efficiency Maximization in
  Wireless-Powered Mobile Edge Computing Networks},'' \emph{{IEEE Trans. Wirel.
  Commun.}}, vol.~19, no.~5, pp. 3170--3184, Feb. 2020.

\bibitem{9771419}
Y.~Cang, M.~Chen, J.~Zhao, T.~Gong, J.~Zhao, and Z.~Yang, ``{Fair Computation
  Efficiency for OFDMA-Based Multiaccess Edge Computing Systems},'' \emph{IEEE
  Commun. Lett.}, vol.~27, no.~3, pp. 916--920, Mar. 2023.

\bibitem{7553459}
K.~Zhang, Y.~Mao, S.~Leng, Q.~Zhao, L.~Li, X.~Peng, L.~Pan, S.~Maharjan, and
  Y.~Zhang, ``{Energy-Efficient Offloading for Mobile Edge Computing in 5G
  Heterogeneous Networks},'' \emph{IEEE Access}, vol.~4, pp. 5896--5907, Aug.
  2016.

\bibitem{9530497}
G.~Shi, Y.~Xiao, Y.~Li, and X.~Xie, ``{From Semantic Communication to
  Semantic-Aware Networking: Model, Architecture, and Open Problems},''
  \emph{{IEEE Commun. Mag.}}, vol.~59, no.~8, pp. 44--50, Aug. 2021.

\bibitem{10024766}
W.~Xu, Z.~Yang, D.~W.~K. Ng, M.~Levorato, Y.~C. Eldar, and M.~Debbah, ``{Edge
  Learning for B5G Networks With Distributed Signal Processing: Semantic
  Communication, Edge Computing, and Wireless Sensing},'' \emph{IEEE J. Sel.
  Topics Sig. Proces.}, vol.~17, no.~1, pp. 9--39, Jan. 2023.

\bibitem{carnap1952outline}
R.~Carnap, Y.~Bar-Hillel \emph{et~al.}, ``{An outline of a theory of semantic
  information},'' 1952.

\bibitem{6004632}
J.~Bao, P.~Basu, M.~Dean, C.~Partridge, A.~Swami, W.~Leland, and J.~A. Hendler,
  ``{Towards a theory of semantic communication},'' in \emph{{Proc. IEEE
  Network Science Workshop}}, West Point, NY, USA, June 2011, pp. 110--117.

\bibitem{9398576}
H.~Xie, Z.~Qin, G.~Y. Li, and B.-H. Juang, ``{Deep Learning Enabled Semantic
  Communication Systems},'' \emph{{IEEE Trans. Signal Process.}}, vol.~69, pp.
  2663--2675, Apr. 2021.

\bibitem{8723589}
E.~Bourtsoulatze, D.~Burth~Kurka, and D.~Gündüz, ``{Deep Joint Source-Channel
  Coding for Wireless Image Transmission},'' \emph{{IEEE Trans. Cogn. Commun.
  Netw.}}, vol.~5, no.~3, pp. 567--579, May 2019.

\bibitem{9450827}
Z.~Weng and Z.~Qin, ``{Semantic Communication Systems for Speech
  Transmission},'' \emph{{IEEE J. Sel. Areas Commun.}}, vol.~39, no.~8, pp.
  2434--2444, Aug. 2021.

\bibitem{9274895}
L.~Cao and S.~Lin, ``{Target Detection Algorithm of Optimized Convolutional
  Neural Network under Computer Vision},'' in \emph{Proc. International
  Conference on Unmanned Systems (ICUS)}, Harbin, China, Nov. 2020, pp.
  923--930.

\bibitem{9416312}
S.~Ji, S.~Pan, E.~Cambria, P.~Marttinen, and P.~S. Yu, ``{A Survey on Knowledge
  Graphs: Representation, Acquisition, and Applications},'' \emph{{IEEE Trans.
  Neural Netw. Learn. Syst.}}, vol.~33, no.~2, pp. 494--514, Feb. 2022.

\bibitem{10032275}
Z.~Yang, M.~Chen, Z.~Zhang, and C.~Huang, ``{Energy Efficient Semantic
  Communication Over Wireless Networks With Rate Splitting},'' \emph{IEEE J.
  Sel. Areas Commun.}, vol.~41, no.~5, pp. 1484--1495, May 2023.

\bibitem{boyd2004convex}
S.~Boyd, S.~P. Boyd, and L.~Vandenberghe, \emph{{Convex optimization}}.\hskip
  1em plus 0.5em minus 0.4em\relax Cambridge university press, 2004.

\bibitem{9449944}
S.~Bi, L.~Huang, H.~Wang, and Y.-J.~A. Zhang, ``{Lyapunov-Guided Deep
  Reinforcement Learning for Stable Online Computation Offloading in
  Mobile-Edge Computing Networks},'' \emph{{IEEE Trans. Wirel. Commun.}},
  vol.~20, no.~11, pp. 7519--7537, Nov. 2021.

\end{thebibliography}
	
\end{document}